\documentclass[10pt]{amsart}
\usepackage{latexsym,amssymb,amsmath,amsthm,amscd,
epic,
graphics,amsxtra
}
\newcommand{\arr}{\longrightarrow}

\renewcommand{\div}{div}

\newcommand{\mx}[1]{\mbox{{#1}}}

\newcommand{\Sym}{{\tt S}}

\DeclareMathOperator{\End}{End}
\DeclareMathOperator{\high}{high}
\DeclareMathOperator{\Hom}{Hom}

\DeclareMathOperator{\low}{low}

\DeclareMathOperator{\tr}{tr}
\def\wt{{\mathrm{wt}}}

\def\CC{{\mathbb{C}}}

\def\ZZ{{\mathbb{Z}}}

\newcommand{\al}{\alpha}

\newcommand{\de}{\delta}

\newcommand{\ph}{\varphi}

\newcommand{\tht}{\theta}

\newcommand{\fg}{\mbox{{\tt g}}}

\newcommand{\sgn}{\mathop{\rm sign}}

\renewcommand{\theequation}%
  {\arabic{section}.\arabic{equation}}
\makeatletter
\renewcommand\section%
 {\@startsection {section}{1}{\z@}%
  {-3.5ex \@plus -1ex \@minus -.2ex}%
    {2.3ex \@plus.2ex}%
       {\normalfont\large\bfseries}}
\@addtoreset{equation}{section}
\makeatother

\newtheorem{proposition}{Proposition}[section]

\newtheorem{Proposition}{Proposition}[section]
\newcommand{\bPr}{\begin{Proposition}}
\newcommand{\ePr}{\end{Proposition}}

\newtheorem{Theorem}[Proposition]{Theorem}
\newcommand{\bTh}{\begin{Theorem}}
\newcommand{\eTh}{\end{Theorem}}

\newtheorem{Lemma}[Proposition]{Lemma}
\newcommand{\bLe}{\begin{Lemma}}
\newcommand{\eLe}{\end{Lemma}}

\newtheorem{Definition}[Proposition]{Definition}
\newcommand{\bDe}{\begin{Definition}}
\newcommand{\eDe}{\end{Definition}}

\newtheorem{Corollary}[Proposition]{Corollary}
\newcommand{\bCo}{\begin{Corollary}}
\newcommand{\eCo}{\end{Corollary}}

\newtheorem{Conjecture}[Proposition]{Conjecture}
\newcommand{\bCj}{\begin{Conjecture}}
\newcommand{\eCj}{\end{Conjecture}}

\theoremstyle{remark}
\newtheorem{remark}[Proposition]{Remark}

\newcommand{\bRe}{\begin{remark}}
\newcommand{\eRe}{\end{remark}}

\newcommand{\bEq}{\begin{equation}}
\newcommand{\eEq}{\end{equation}}

\newcommand{\bEz}{\begin{equation*}}
\newcommand{\eEz}{\end{equation*}}

\newcommand{\bEa}{\begin{eqnarray}}
\newcommand{\eEa}{\end{eqnarray}}

\newcommand{\bEaz}{\begin{eqnarray*}}
\newcommand{\eEaz}{\end{eqnarray*}}

\newcommand{\bAr}{\begin{array}}
\newcommand{\eAr}{\end{array}}

\newcommand{\bN}{\begin{enumerate}}
\newcommand{\eN}{\end{enumerate}}

\newcommand{\bD}{\begin{description}}
\newcommand{\eD}{\end{description}}

\newcommand{\prf}{{\sl Proof.}}
\newcommand{\pth}{{\sl Proof of the theorem}}

\newcommand{\epf}{$\Box$}

\begin{document}


\title[Morphisms of Verma modules over $E (5,10)$.]
{Morphisms of Verma modules over exceptional Lie superalgebra $E (5,10)$.}

\author{Alexei Rudakov}

\begin{abstract}
In this paper we define the degree of a morphism between
(generalized) Verma modules over a graded Lie superalgebra and
construct series of morphisms of various degrees between
(generalized) Verma modules over the exceptional
infinite-dimensional linearly compact simple Lie superalgebra
$E(5,10)$. We prove that all such morphisms of degree 1 are found.
\end{abstract}
\maketitle

{\raggedleft \mbox{{\small {\it To Victor Kac for the birhtday }}}\\}
\section{Introduction.}

The paper continues the study of representations of the exceptional
infinite-dimensional linearly compact simple Lie superalgebra
$E(5,10)$ that has begun in \cite{KR3}. Here we deal with morphisms
between (generalized) Verma modules. We define a degree of such a
morphism and notice  that the morphisms described in \cite{KR3} are
morphisms of degree 1.
In the paper we prove that there are no more morphisms of degree 1.
But we find morphisms of degrees 2 and 3 when considering products
of morphisms of degree 1.

It takes some efforts to find morphisms of degree 4. We study their
properties and use them to construct also some morphisms of degree
5. We show that the picture of complexes of Verma modules over
$E(5,10)$ given in \cite{KR3} extends naturally with the morphisms
of degree 4.

At the moment it is not clear if there are more morphisms of degree
$\geq 3$, but we know that there are no  other morphisms of degree 2. We
will include the proof of this fact in a subsequent paper.



\section{Verma modules.}
\label{sec:1}

We keep the notations from \cite{KR2} (see also \cite{KR1}, \cite{R}).

 Let $L = \oplus_{j \in \ZZ} \fg_j$ be a $\ZZ$-graded Lie
superalgebra by finite-dimensional vector spaces.  Let
\begin{displaymath}
  L_- = \oplus_{j<0}\,\, \fg_j , \,\,\, L_+ = \oplus_{j>0} \,\,\fg_j ,
  \,\,\, L_0 = \fg_0 + L_+ \, .
\end{displaymath}
As usual we denote by $U=U(L)$ the universal enveloping algebra of
$L$ and similarly $U_{0}=U{(L_0)}$,  $U_{-}=U{(L_{-})}$.

Let us notice that the grading extends to the enveloping algebras, in particular
to $U_-$. It is convenient to change the sign of the degree when considering this
grading on $U_-$. We shall call it the natural grading of $U_-$.

Given a $\fg_{{0}}$-module $V$, we extend it to a $L_{{0}}$-module by
letting $L_+$ act trivially, and define the induced $L$-module
\begin{displaymath}
  M(V) = U\otimes_{U_0} V \, .
\end{displaymath}
We shall use the vector space isomorphism
\begin{displaymath}
  M(V) = U_- \otimes_{\CC} V \, .
\end{displaymath}
\bDe Let $V$ be a $\fg$-module.
The $L$-module $M(V)$ is called a (generalized) \emph{Verma module}
(associated to $V$).
When we want to emphasis that $V$ is a finite-dimensional irreducible  $\fg$-module,
we call the Verma module $M(V)$ a  \emph{minimal} Verma module.\\
A minimal Verma module is called \emph{non-degenerate}
if it is irreducible and \emph{degenerate} if it is not irreducible.
\eDe
Let $A$ and $B$ be two $\fg_0$-modules and let $\Hom (A,B)$ be the
$\fg_0$-module of linear maps from $A$ to $B$.  The following
proposition will be used to construct morphisms between the
$L$-modules $M(A)$ and $M(B)$.

\begin{proposition}
\label{prop:1.1}
Let $\varphi:M(A) \to M(B)$  be a morphism of $L$-modules and\\
$\Phi \in U_- \otimes_{\CC}\Hom (A,B)$ be such that $\varphi (1
\otimes a) = \Phi (a)$.
Then
$\Phi$ has the properties:\\
\begin{subequations}\label{eq:1.1}
\begin{eqnarray}
\label{eq:1.1a}
 &\, v\,\Phi (a)=\Phi (u\,a) \text{ for } v \in \fg_{{0}}\, ,\\
 \label{eq:1.1b}
 &\, v\,\Phi (a)=0 \text{ for } v \in L_+ \,.
\end{eqnarray}
\end{subequations}
and the morphism $\varphi$ is determined by the rule
\begin{equation}
\label{eq:1.2}
 \varphi (u \otimes a) =  u\,\Phi (a) \, .
\end{equation}
Moreover for any  $\Phi$ with the properties (\ref{eq:1.1a}),
(\ref{eq:1.1b}) a morphism $\varphi : M(A) \to M(B)$ is uniquely
defined.
\end{proposition}

\begin{proof}
If $\varphi : M(A) \to M(B)$ is a morphism of $L$-modules
  then for $u \in U (L_0)$,
  \begin{displaymath}
u\,\varphi (1\otimes a)= \varphi (u \otimes a) = \varphi (1 \otimes
u\,a) \, .
  \end{displaymath}
Properties (\ref{eq:1.1a}) and (\ref{eq:1.1b}) follow.

Now given $\Phi \in U_- \otimes_{\CC}\Hom (A,B))$ with the
properties  (\ref{eq:1.1a}),(\ref{eq:1.1b}) we may wish to use
(\ref{eq:1.2}) to define $\varphi$, but in order to conclude that
$\varphi$ is well-defined we have to check for any $u \in U$, $v\in
U_0$ the equality
\begin{displaymath}
\varphi (uv \otimes a)= \varphi (u \otimes v\,a)\, .
\end{displaymath}
Clearly it is equivalent to an equality
\begin{displaymath}
v\,\Phi(a)= \Phi (v\,a)\,,
\end{displaymath}
and it is sufficient to consider cases: $v\in \fg_{{0}}$ and $v\in
L_+$. For the first case the equality is exactly the same as the
property (\ref{eq:1.1a}). In the second case we have $v\,a=0$
because $L_+$ act trivially and we come to the property
(\ref{eq:1.1b}).
\end{proof}

Let us denote the morphism defined in the proposition as $\Phi=\ph|_A$
and call it the restriction of $\ph$.
\bDe\label{d:deg-m}
We say that a morphism $\ph:M(A) \to M(B) $ has  degree $k$ when
\[
\Phi =\ph|_A = \sum_i u_i \otimes \ell_i\,, \text{ where } u_i
\in U_-\,, \,\ell_i \in \Hom (A,B)\,
\]
and $\deg u_i = k $ for all $i$.
\eDe
We shall permit ourselves to denote the morphism of Verma modules and its restriction
by the same letter when it does not lead to confusion.
\begin{remark}
  \label{rem:1.2}
If $L_0$ is generated by $\fg_0$ and a subset of $T \subset L_+$,
then conditions (\ref{eq:1.1}) are equivalent to
\begin{subequations}
\begin{eqnarray}
  \label{eq:1.4a}
\fg_0 \cdot \Phi &=& 0 \, \\
\label{eq:1.4b}
   t \, \Phi (a)&=& 0 \,\,\hbox{ for all }\, t \in T \,,\, a\in A\, .
\end{eqnarray}
\end{subequations}
The "dot" in (\ref{eq:1.4a}) denotes the action of $\fg_0$ on the
tensor product of $\fg_0$-modules $U_-$ and $\Hom (A,B)$. Usually it
gives a hint to a possible choice of $\Phi$ and may be checked by
general invariant-theoretical considerations.  The condition
(\ref{eq:1.4b}) is often more difficult to check.
\end{remark}

\begin{remark}
  \label{rem:1.3}
We can view $M(V)$ also as the induced $(L_- \oplus \fg_0)$-module:
$U(L_- \oplus \fg_0) \otimes_{U(\fg_0)}V$.  Then condition
(\ref{eq:1.4a}) on $\Phi = \sum_i u_i \otimes \ell_i$, where $u_i
\in U(L_- \oplus \fg_0)$, $\ell_i \in \Hom (A,B)$, suffices in order
(\ref{eq:1.2}) to give a well-defined morphism of
$(L_- \oplus \fg_0)$-modules.  One can also replace $\fg_0$ by any of its
subalgebras.
\end{remark}



\section{Lie superalgebra $E(5,10)$.}
\label{sec:2}

Recall some standard notation:
\begin{displaymath}
  W_n = \{ \sum^n_{j=1} P_i (x) \partial_i\,\,
| \,\, P_i \in \CC [[x_{{1}} , \ldots ,x_n ]], \,\,
  \partial_i \equiv \partial / \partial x_i \}
\end{displaymath}
denotes the Lie algebra of formal vector fields in $n$
indeterminates:
\begin{displaymath}
  S_n = \{ D = \sum P_i \partial_i  \, \, |\,\,
  \div D \equiv \sum_i \partial_i P_i =0 \}
\end{displaymath}
denotes the Lie
  subalgebra of divergenceless formal vector fields; $\Omega^k (n)$
  denotes the associative algebra of formal differential forms of
  degree $k$ in $n$ indeterminates, $\Omega^k_{c\ell} (n)$ denoted
  the subspace of closed forms.

The Lie algebra $W_n$ acts on $\Omega^k (n)$ via Lie derivative $D
\to L_D$.  Given $\lambda \in \CC$ one can define the twisted
action:
\begin{displaymath}
  D \,\omega = L_D \,\omega + \lambda (\div D)\, \omega \, .
\end{displaymath}
The $W_n$-module thus obtained is denoted by $\Omega^k
(n)^{\lambda}$.  Recall the following isomorphism of $W_n$-modules
\begin{equation}
  \label{eq:2.2}
  W_n \simeq \Omega^{n-1} (n)^{-1} \, .
\end{equation}
It is obtained by mapping a vector field $D \in W_n$ to the
{\footnotesize{$(n-1)$}}-form $\iota_D (dx_{{1}} \wedge \ldots
\wedge dx_n)$.  Note that (\ref{eq:2.2}) induces an isomorphism of
$S_n$-modules:
\begin{equation}
  \label{eq:2.3}
  S_n \simeq \Omega^{n-1}_{c\ell} (n) \, .
\end{equation}

Recall that the Lie superalgebra $E(5,10) = E(5,10)_{\bar{0}}
\dotplus E (5,10)_{\bar{1}}$ is constructed as follows \cite{K},
\cite{CK}:
\[
E(5,10)_{\bar{0}} = S_5, \quad E(5,10)_{\bar{1}} =
\Omega^2_{c\ell}(5).
\]\medskip
To describe brackets consider an algebra $\widetilde{E}(5,10) =
\widetilde{E}(5,10)_{\bar{0}} \,\dotplus \,\widetilde{E}(5,10)_{\bar{1}}
= W_5 \,\dotplus \,\Omega^2(5)$ where the subalgebra
$ W_5 $ acts on
$\widetilde{E}(5,10)_{\bar{1}}$ via the Lie derivative, but for
$\,\,\omega_2,\omega'_2 \in \widetilde{E}(5,10)_{\bar{1}}$ the
brackets are
\[
\,[\omega_2 , \omega'_2] = \omega_2 \wedge \omega'_2 \in \Omega^4
(5) \stackrel{\sim}{\to} W_5\, \quad \text{(see \,
 (\ref{eq:2.2}) )} \, .
\]
Now $E(5,10)$ is a subalgebra of $\widetilde{E}(5,10)$.
Let as note that $\widetilde{E}(5,10)$ is not a Lie superalgebra, the Jacobi identity
is no longer true in this larger algebra, but true in $E(5,10)$\,.

As in \cite{KR2}  we use for the odd
elements of $E(5,10)$ the notation $d_{ij} = dx_i \wedge dx_j$
$(i,j=1,2,\ldots ,5)$; recall that we have the following
commutation relation $(f,g \in \CC [[ x_1 , \ldots , x_5 ]])$:
\begin{displaymath}
  [f d_{jk}, g \ d_{\ell m}] = \epsilon_{ij k \ell m}\,f g\, \partial_i
  \, ,
\end{displaymath}
where $\epsilon_{ijk\ell m}$ is the sign of the permutation
$ijk\ell m$ if all indices are distinct and $0$ otherwise.

Recall that the Lie superalgebra $L=E(5,10)$ carries a unique
consistent irreducible $\ZZ$-gradation $L = \oplus_{j \geq
  -2}\, \,\fg_j$.  It is defined by:
\begin{displaymath}
  \deg x_i =2 =-\deg \partial_i \, , \,\, \deg d_{ij} =-1 \, .
\end{displaymath}
One has:  $\fg_0 \simeq s\ell_5 (\CC)$ and the $\fg_0$-modules
occurring in the $L_-$ part are:
\begin{eqnarray*}
  \fg_{-1} &=& \langle d_{ij} \,\, | i,j=1,\ldots ,5 \rangle \simeq
     \Lambda^2 \CC^5 \, ,\\
     fg_{-2} &=& \langle \partial_i \,\, | i=1,\ldots ,5 \rangle
        \simeq \CC^{5*} \, .
\end{eqnarray*}
Recall also that $\fg_1$ consist of closed $2$-forms with linear
coefficients, that $\fg_1$ is an irreducible $\fg_0$-module and $\fg_j
 = [\,\fg_1\,[\ldots ]]= \fg_1^j$ for $j \geq 1$.

%


\section{Degenerate Verma modules and morphisms of degree 1.}
\label{sec:4}

We take for the Borel subalgebra of $\fg_0\simeq s\ell_5$ the subalgebra of the
vector fields
\[
\{ \sum_{i \leq j} a_{ij} (x_i \partial_j)\, |\,
a_{ij} \in \CC \, , \, \tr (a_{ij})=0 \} =
\langle    x_i \partial_j , \,\,{i \leq j}, \, h_i=x_i \partial_i - x_{i+1}\partial_{i+1}
 \rangle\,.
\]
We denote by
$F(n_1,n_2,n_3,n_4)$ the finite-dimensional irreducible
$\fg_0$-module with highest weight $(n_1,n_2,n_3,n_4)$.  Let
\begin{displaymath}
  M(n_1,n_2,n_3,n_4) = M(F(n_1,n_2,n_3,n_4))
\end{displaymath}
denote the corresponding generalized Verma module over $E(5,10)$.

Let us repeat a conjecture from \cite{KR3} (where it contains a misprint)
\begin{Conjecture}
  \label{conj:4.1}
The following is a complete list of degenerate Verma modules
over $E(5,10)$:
\begin{displaymath}
  M(m,n,0,0) , \,\, M(m,0,0,n), \,\hbox{ and }  \, M(0,0,m,n)\quad (\text{for any } m,n \in \ZZ_+).
\end{displaymath}
\end{Conjecture}

In this section we construct three series of morphisms of degree one of
Verma modules which shows, in particular, that all
modules from the list given by Conjecture~\ref{conj:4.1} are indeed
degenerate.

We let (see \cite{KR3}, but our notations differ slightly):
\begin{displaymath}
  S_A = S (\CC^5 + \Lambda^2 \CC^5) \, ,\,
  S_B = S (\CC^5 + \CC^{5*}) \, , \,
  S_C = S(\CC^{5*} + \Lambda^2 \CC^{5*}) \,.
\end{displaymath}
Denote by $z_i$ $(i=1,\ldots ,5)$ the standard basis of $\CC^5$
and by $x_{ij}=-x_{ji}$ $(i,j=1, \ldots ,5)$ the standard
basis of $\Lambda^2 \CC^5$.
Let $z^*_i$ and $x^*_{ij} =-x^*_{ji}$
be the dual bases of $\CC^{5*}$ and $\Lambda^2\CC^{5*}$,
respectively.  Then $S_A$ is the polynomial algebra in
$15$~indeterminates $x_i$ and $x_{ij}$, $S_B$ is the polynomial
algebra in $15$~indeterminates $z^*_i$ and $x^*_{ij}$ and $S_C$
is the polynomial algebra in $10$~indeterminates $z_i$ and
$z^*_i$.

Given two irreducible $\fg_0$-modules $E$ and $F$, we denote by $(E
\otimes F)_{\high}$ the highest irreducible component of the
$\fg_0$-module $E \otimes F$.  If $E = \oplus_i E_i$ and $F=
\oplus_jF_j$ are direct sums of irreducible $\fg_0$-modules, we let
$(E \otimes F)_{\high} = \oplus_{i,j} (E_i \otimes F_j)_{\high}$.
If $E$ and $F$ are again irreducible $\fg_0$-modules, then $S (E
\oplus F)= \oplus_{m,n \in \ZZ_+} S^m E \otimes S^nF$, and we let
$S_{\high} (E \oplus F) = \oplus_{m,n \in \ZZ_+} (S^m E \otimes
S^nF)_{\high}$.  We also denote by $S_{\low} (E \oplus F)$ the
$\fg_0$-invariant
complement to $S_{\high}(E \oplus F)$.

It is easy to see that we have as $\fg_0$-modules:
\begin{eqnarray*}
  S_{A,\high} & \simeq & \oplus_{m,n \in \ZZ_+}
      F (m,n,0,0) \, , \\
  S_{B,\high}  & \simeq &\oplus_{m,n \in \ZZ_+}
       F (m,0,0,n) \,, \\
  S_{C,\high}  & \simeq & \oplus_{m,n \in \ZZ_+}
       F(0,0,m,n) \,.
\end{eqnarray*}
We are going to construct morphisms  $\triangledown_X \in \End_L(M(S_{X,\high}))$ of
degree 1
for  $X=A$, $B$ or $C$.
Let us start with the following operators:
\begin{equation}\label{eq:4-tr}
  \triangledown_X = \sum^5_{i,j=1} d_{ij} \otimes \theta^X_{ij}
  \, ,
\end{equation}
where $ \theta^X_{ij}\in \Hom(S_X,\,S_X)$, $X=A$, $B$ or $C$. Namely we take
\begin{displaymath}
  \theta^A_{ij} = \frac{d}{dx_{ij}} \, ,\quad
  \theta^B_{ij} =z^*_i \frac{d}{dz_j}-z^*_j\frac{d}{dz_i} \,, \quad
  \theta^C_{ij} = x^*_{ij} \, .
\end{displaymath}
In order to get the morphisms of $S_{X,\high}$ we need following more
specific realizations of these spaces.
\bPr\label{prop:4.1c}
The $\fg_0$-module $S_{C,\high}$ is equal to a factor of the polynomial ring $S_C$
by the ideal $S_{C,\low}$ generated by
the relations:
\begin{subequations}
\begin{eqnarray}
    x^*_{ab} x^*_{cd}
      -   x^*_{ac} x^*_{bd}
      +   x^*_{ad} x^*_{bc}
&=0 \, \text{ for $a,b,c,d=1,\ldots,5$}\,,\\
    x^*_{ab} z^*_{c}
      -   x^*_{ac} z^*_{b}\label{eq:4.1c}
      +   x^*_{bc} z^*_{a}
 &=0 \, \text{ for $a,b,c=1,\ldots,5$}\,.
\end{eqnarray}
\end{subequations}
\ePr
This follows from the fact that the orbit of the sum of highest weight vectors in
$F(0,0,0,1)\oplus F(0,0,1,0)$ is a spherical variety.
\bPr\label{prop:4.1b} The $\fg_0$-module
$S_{B,\high}$ is equal to a factor of the polynomial ring $S_B$ by
the ideal $S_{B,\low}$ generated by the relation:
\begin{eqnarray}
  \label{eq:4.1b}
   \sum_{i=1,..,5} z_{i} z^*_{i}
       &=0 \, .
\end{eqnarray}
\ePr
Similarly we use the fact that the orbit of the sum of highest weight vectors in
$F(1,0,0,0)\oplus F(0,0,0,1)$ is a spherical variety.
\bPr\label{prop:4.1a}
The $\fg_0$-module $S_{A,\high}$ is equal to the subspace of all polynomials
 $f\in S_{A}$ that satisfy the following equations:
\begin{subequations}\label{eq:4.1A}
\begin{eqnarray}
  \label{eq:4.1a}
\left( \frac{d}{dx_{ab}}\frac{d}{dx_{cd}}
     - \frac{d}{dx_{ac}}\frac{d}{dx_{bd}}
     + \frac{d}{dx_{ad}}\frac{d}{dx_{bc}}
\right) f=0&\,  \text{ for $a,b,c,d=1,\ldots,5$}\,,\\   
\left( \frac{d}{dx_{ab}}\frac{d}{dz_{c}}
     - \frac{d}{dx_{ac}}\frac{d}{dz_{b}}
     + \frac{d}{dx_{bc}}\frac{d}{dz_{a}}
\right) f=0&\,  \text{ for $a,b,c=1,\ldots,5$\,.}
\end{eqnarray}
\end{subequations}
\ePr
We notice that the realization of $S_C$ as differential
operators on $S_A$ makes the perfect duality.
The proposition means that $S_{A,\high}$ coincides with the orthogonal complement
to $S_{C,\low}$, which is clear because $S_{A,\high}$ is dual to $S_{C,\high}$.

\bCo\label{cor:4.1}
$(a)$ The operator $\triangledown_A$ preserves the subspace $S_{A,\high}$ of $S_{A}$ and
therefore defines a map $\triangledown_A \in \End(S_{A,\high})$.\\
$(b)$ The operator $\triangledown_B$ preserves the ideal $S_{B,\low}$ of $S_{B}$ and
thus determines a map of the factor rings $\triangledown_B \in \End(S_{B,\high})$.\\
$(c)$ The operator $\triangledown_C$ preserves the ideal $S_{C,\low}$ of $S_{C}$ and
determines a map of the factor rings $\triangledown_C \in \End(S_{C,\high})$.
\eCo
The statements $(a)$ and $(c)$  are evident. For $(b)$ we notice that
$\triangledown_B$ acts as a derivative, thus it is enough to check that it annihilates the
generator of the ideal $S_{B,\low}$ given by (\ref{eq:4.1b}), which is
straightforward.
\medskip

To keep with the notations in \cite{KR3} we let:
\begin{displaymath}
  V_A = S_{A,\high} \, , \quad V_{B (\hbox{resp. }C)} =
  S_{B (\hbox{resp. }C)}/ S_{B(\hbox{resp. }C){,\low}}\simeq S_{B (\hbox{resp. }C){,\high}}.
\end{displaymath}
\begin{Theorem}
  \label{th:4.1}
$(a)$~~The operators $\triangledown_X$ define $E(5,10)$-morphisms
$M(V_X) \to M(V_X) \,\,\\ (X=A,B \hbox{ or } C)$.

$(b)$~~Morphism $\triangledown_X$ restricted to $M=M(n_1,n_2,n_3,n_4)$ is trivial
iff $X=A$ and \\
$M=M(m,0,0,0)$,  or $X=B$ and $M=M(0,0,0,n)$.

$(c)$~~The non-zero morphisms $\triangledown_X$ are morphisms of degree 1.
\end{Theorem}

\prf\/
It is immediate to see that
$\fg_0 \cdot \triangledown_X =0$. By Remark~\ref{rem:1.3} we conclude that
there exist the corresponding $(U_-\oplus\fg_0)$-morphisms $M(S_X) \arr M(S_X)$ that we
permit ourselves to denote by the same symbols $\triangledown_X$.
These are evidently morphisms of degree 1.

In order to apply Proposition~\ref{prop:1.1} and get $E(5,10)$-morphisms
of modules $S_{X,\high}$
we need to check that
\begin{equation}
  \label{eq:4-check}
  \fg_1 \cdot \triangledown_X (s)=0 \, \text{ for $s\in S_{X,\high}$}.
\end{equation}
Now in order to check (\ref{eq:4-check}) we may use Remark~\ref{rem:1.2}
with $t=x_5d_{45}$. Namely
\bEq\label{eq:4.2}
\begin{aligned}
(x_5d_{45})&\sum^5_{i,j=1} d_{ij} \otimes \theta^X_{ij}(s) =\\
&=[x_5d_{45},\,d_{12}]\otimes \theta^X_{12}(s)+
[x_5d_{45},\,d_{13}]\otimes \theta^X_{13}(s)+
[x_5d_{45},\,d_{23}]\otimes \theta^X_{23}(s)+0\\
&=(x_5\partial_3)\otimes \theta^X_{12}(s)-
(x_5\partial_2)\otimes \theta^X_{13}(s)+
(x_5\partial_1)\otimes \theta^X_{23}(s)\,.
\end{aligned}
\eEq
It is not difficult to check that for $X=A$ the right
hand side is equal to zero modulo
relations (\ref{eq:4.1A}).

Now let us consider $X=B$. Here for $s\in S_B$
\begin{displaymath}
(x_5\partial_c)\theta^B_{ab}(s)=
\left(z_5z^*_a \frac{d}{dz_b}\frac{d}{dz_c}-
z^*_az^*_c\frac{d}{dz_b}\frac{d}{dz^*_5}
-z_5z^*_b \frac{d}{dz_a}\frac{d}{dz_c}+
z^*_bz^*_c\frac{d}{dz_a}\frac{d}{dz^*_5}\right)(s)\,.
\end{displaymath}
This immediately gives us zero at the right hand side of (\ref{eq:4.2}).

When $X=C$ we have
\begin{displaymath}
(x_5\partial_c)\theta^B_{ab}(s)=
-z^*_c x^*_{ab}\frac{d}{dz^*_5}(s)\,,
\end{displaymath}
thus the right hand side of (\ref{eq:4.2}) is zero modulo the relations
(\ref{eq:4.1c}).\epf\\
The non-zero maps $\triangledown_X$ are illustrated in
Figure~\ref{fig:2}.


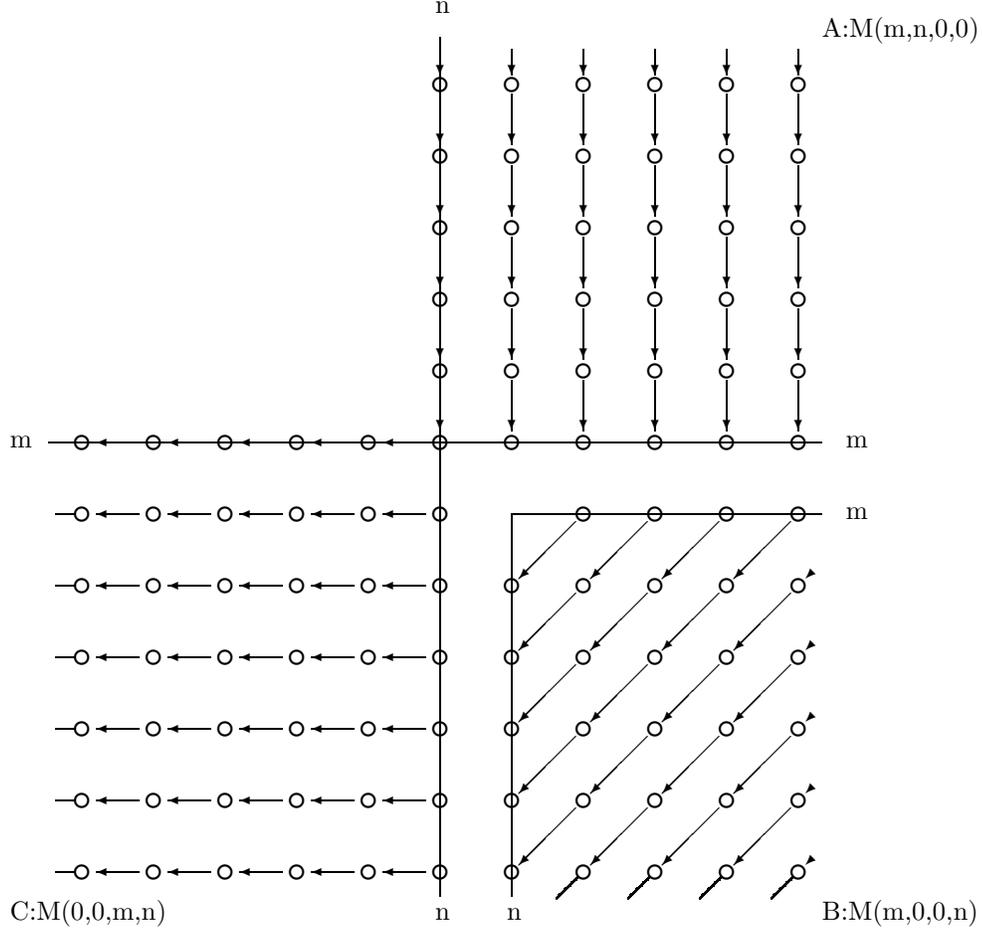
\begin{figure}[htbp]
  \begin{center}
    \leavevmode
  \setlength{\unitlength}{0.25in}
\begin{picture}(21,19)


\put(18,17.5){A:M(m,n,0,0)}
\put(18,-1){B:M(m,0,0,n)}

\put(1,17.5){}
\put(1,-1){C:M(0,0,m,n)}


\put(10,17.5){\line(0,-1){8.5} }
\put(9.9,18){n}

\put(10,-.5){\line(0,1){9.5} }
\put(9.9,-1){n}

\put(11.5,-.5){\line(0,1){8} }
\put(11.4,-1){n}

\put(11.5,7.5){\line(1,0){6.5} }
\put(18.5,7.4){m}

\put(1.8,9){\line(1,0){8.2} }
\put(1,8.9){m}

\put(10,9){\line(1,0){8} }
\put(18.5,8.9){m}


\thicklines
\multiput(2.5,0)(0,1.5){7}{\circle{.25} }
\multiput(4,0)(0,1.5){7}{\circle{.25} }  
\multiput(5.5,0)(0,1.5){7}{\circle{.25} }
\multiput(7,0)(0,1.5){7}{\circle{.25} }  
\multiput(8.5,0)(0,1.5){7}{\circle{.25} }

\multiput(10,0)(0,1.5){12}{\circle{.25} }  
\multiput(11.5,0)(0,1.5){5}{\circle{.25} }
\multiput(11.5,9)(0,1.5){6}{\circle{.25}} 
\multiput(13,0)(0,1.5){12}{\circle{.25} }  
\multiput(14.5,0)(0,1.5){12}{\circle{.25} }
\multiput(16,0)(0,1.5){12}{\circle{.25} }  
\multiput(17.5,0)(0,1.5){12}{\circle{.25} }

\thinlines

\drawline(1.95,0  )(2.35,0  )
\drawline(1.95,1.5)(2.35,1.5)
\drawline(1.95,3  )(2.35,3  )
\drawline(1.95,4.5)(2.35,4.5)
\drawline(1.95,6  )(2.35,6  )
\drawline(1.95,7.5)(2.35,7.5)
\drawline(1.95,9  )(2.35,9  )


\drawline(12.45,-.55)(12.9,-.1)
\drawline(13.95,-.55)(14.4,-.1)
\drawline(15.45,-.55)(15.9,-.1)
\drawline(16.95,-.55)(17.4,-.1)

\multiput(18.2,.7)(0,1.5){5}{\vector(-1,-1){.55}}
\multiput(10,17.25)(1.5,0){6}{\vector(0,-1){.55}}


\multiput(10,16.3)(1.5,0){6}{\vector(0,-1){1.05}}
\multiput(10,14.8)(1.5,0){6}{\vector(0,-1){1.05}}
\multiput(10,13.3)(1.5,0){6}{\vector(0,-1){1.05}}
\multiput(10,11.8)(1.5,0){6}{\vector(0,-1){1.05}}
\multiput(10,10.3)(1.5,0){6}{\vector(0,-1){1.05}}

\multiput(3.7,9  )(1.5,0){5}{\vector(-1,0){0.9}}
\multiput(3.7,7.5)(1.5,0){5}{\vector(-1,0){0.9}}
\multiput(3.7,6  )(1.5,0){5}{\vector(-1,0){0.9}}
\multiput(3.7,4.5)(1.5,0){5}{\vector(-1,0){0.9}}
\multiput(3.7,3  )(1.5,0){5}{\vector(-1,0){0.9}}
\multiput(3.7,1.5)(1.5,0){5}{\vector(-1,0){0.9}}
\multiput(3.7,0  )(1.5,0){5}{\vector(-1,0){0.9}}

\multiput(12.85,7.35)(1.5,0){4}{\vector(-1,-1){1.2}}
\multiput(12.85,5.85)(1.5,0){4}{\vector(-1,-1){1.2}}
\multiput(12.85,4.35)(1.5,0){4}{\vector(-1,-1){1.2}}
\multiput(12.85,2.85)(1.5,0){4}{\vector(-1,-1){1.2}}
\multiput(12.85,1.35)(1.5,0){4}{\vector(-1,-1){1.2}}
\end{picture}
\vspace{3ex}
  \caption{Morphisms of degree 1.}
    \label{fig:2}
  \end{center}
\end{figure}
 The nodes in the quadrants $A$, $B$ and $C$
represent generalized Verma modules $M(m,n,0,0)$, $M(0,0,m,n)$
and $M(m,0,0,n)$, respectively.  The arrows represent the
$E(5,10)$-morphisms $\triangledown_X$, $X = A$, $B$ or $C$ in the
respective quadrants.

\medskip

\bPr
$(\triangledown_X)^2=0 \,\, (X=A,B \hbox{ or } C)$.
\ePr
It is not difficult to check the following lemma.
\begin{Lemma}
  \label{lem:4.2}
Equation $(\triangledown_X)^2 =0$  for $\triangledown_X \in \End_L(M(V_X))$ is
equivalent to the system of equations :
\begin{displaymath}
  \theta_{ab}^X \theta_{cd}^X-\theta_{ac}^X\theta_{bd}^X
  + \theta_{ad}^X\theta_{bc}^X =0 \, \text{ for $a,b,c,d=1,\ldots ,5$}.
\end{displaymath}
\end{Lemma}
For $X=A$ the equations follow from Proposition~\ref{prop:4.1a}, for $X=C$
the equations follow immediately from Proposition~\ref{prop:4.1c}. In the case
$X=B$ denote
\begin{displaymath}
[ab/cd]=x_a^* x_b^*\frac{d}{dx_{c}}\frac{d}{dx_{d}}\,.
\end{displaymath}
Then we have
\begin{displaymath}
\theta_{ab}^B\theta_{cd}^B = [ac/bd]-[ad/bc]-[bc/ad]+[bd/ac]\,,
\end{displaymath}
and similarly
\begin{displaymath}
\begin{aligned}
\theta_{ac}^B\theta_{bd}^B = [ab/cd]-[ad/bc]-[bc/ad]+[cd/ab]\,,\\
\theta_{ab}^B\theta_{cd}^B = [ab/cd]-[ac/bd]-[bd/ac]+[cd/ab]\,.
\end{aligned}
\end{displaymath}
The equation follows.

\begin{Theorem}
  \label{th:4.2}
All the morphisms of degree 1 are the following:\smallskip
\\ \smallskip
$(a)$~~$\triangledown_A:M(n_1,n_2,0,0) \to M(n_1,n_2-1,0,0)$ \quad
for $n_1\geq 0$, $n_2>0$;\\ \smallskip
$(b)$~~$\triangledown_B:M(n_1,0,0,n_4)\to M(n_1-1,0,0,n_4+1)$
for $n_1> 0$, $n_4\geq 0$; \\ \smallskip
$(c)$~~$\triangledown_C:M(0,0,n_3,n_4)\to M(0,0,n_3,n_4+1)$ \quad
for $n_3\geq 0$, $n_4\geq 0$.
\end{Theorem}
\begin{remark} In short these are the morphisms represented
in Figure~\ref{fig:2}, i.e. all the non-zero morphisms $\triangledown_X$.
\end{remark}
\prf\/
A morphism $\ph : M(\bar{n}')\to M(\bar{n}'')$ of degree 1 is defined
according to Proposition~\ref{prop:1.1} by $\Phi$ of the form
\begin{equation}\label{eq:4.3}
  \Phi = \sum^5_{i,j=1} d_{ij} \otimes \theta_{ij}
  \, ,
\end{equation}
Because of \eqref{eq:1.4a} the basis $\{\theta_{ij}\}$ is dual to
$\{d_{ij}\}$ and that means
\begin{equation}\label{eq:4.4}
  (x_a\partial_b)\, \theta_{ij}= - \de_{ai}\theta_{bj}-\de_{aj}\theta_{ib}
  \, ,
\end{equation}
The equation \eqref{eq:1.4b} where
we consider $t= x_qd_{pq}$ gives (the calculations are similar
to those in \eqref{eq:4.2} ):
\bEz
\begin{aligned}
(x_qd_{pq})&\sum^5_{i,j=1} d_{ij} \otimes \theta_{ij}(s) =\\
&= (x_q\partial_c)\otimes \theta_{ab}(s)-
(x_q\partial_b)\otimes \theta_{ac}(s)+
(x_q\partial_a)\otimes \theta_{bc}(s)=0\,,
\end{aligned}
\eEz
or
\bEq\label{eq:4.5}
(x_q\partial_c) \theta_{ab}(s)+
(x_q\partial_b) \theta_{ca}(s)+
(x_q\partial_a) \theta_{bc}(s)=0\,,
\eEq
for different $a,b,c,q$.

Choose $s$ to be the highest weight vector in $F(\bar{n}')$. \\
Suppose first that $\theta_{45}(s)\neq 0$  and $ a,q \leq 3$, $b=4,\,c=5$.
Then \eqref{eq:4.5} and \eqref{eq:4.4} imply
\bEz
(x_q\partial_a)\otimes \theta_{45}(s)= -\tht_{a4}(x_q\partial_5)s
                                         -\tht_{5a}(x_q\partial_4)s =0\,.
\eEz
Therefore the weight $\wt(\theta_{45}(s))$ has the form $(0,0,m,n)$ for some
$m,n\geq 0$.
We conclude that
\bEz
\bar{n}'=\wt(s)=(0,0,m+1,n)\,.
\eEz
This clearly implies $\ph=\triangledown_C$.

We assume $\theta_{45}(s)= 0$ for the following.\\
If we apply $x_c\partial_q$ to \eqref{eq:4.5} we get
\bEq \label{eq:4.6}
\theta_{ab}(h_{cq}s)+
\theta_{ca}(x_c\partial_b) (s)-\theta_{qa}(x_q\partial_b) (s)+
\theta_{bc}(x_c\partial_a) (s)-\theta_{bq}(x_q\partial_a) (s)=0\,.
\eEq
For $a=3, b=5, c=2,q=4$ it becomes
\bEz
h_{24}\theta_{35}(s)=
-\theta_{23}(x_2\partial_5) (s)+\theta_{43}(x_4\partial_5) (s)-
\theta_{52}(x_2\partial_3) (s)+\theta_{54}(x_4\partial_3) (s)=
(x_4\partial_3)\theta_{54} (s)=0
\,.
\eEz
Suppose that $\theta_{35}(s)\neq 0$.
Then $\theta_{35}(s)$ is a highest weight vector and its weight has
the form $\wt(\theta_{35})(s)=(m,0,0,n)$ for some $m,n\neq 0$.
But when $\wt(s)=(m,0,0,n)-(0,1,-1,1)$ is not dominant.

So we have $\theta_{35}(s)=\theta_{45}(s)=0$.\\
If $\tht_{25}(s)\neq 0$, then it is a highest weight vector and
similarly we may use \eqref{eq:4.6}
for $a=2$, $b=5$, $c=1$, $q=4$. We come to $h_{14}\theta_{25}(s)=0$, which
gives an impossible weight for $\tht_{25}(s)$. Thus $\theta_{25}(s)=0$.

Substitute $a=1,b=5, 1<c,q<5$ in \eqref{eq:4.6}
\bEz
h_{cq}\theta_{15}(s)=
-\theta_{1c}(x_c\partial_5) (s)+\theta_{1q}(x_q\partial_5) (s)+
(x_c\partial_1)\theta_{c5} (s)-(x_q\partial_1)\theta_{q5} (s)= 0
\,.
\eEz
We see that if $\theta_{15}(s)\neq 0$, then it is the
highest weight vector with the weight $\wt(\theta_{15}(s))=(m,0,0,n)$.
It is easy to conclude that $\ph=\triangledown_B$ in this case.

We have all $\tht_{i5}=0$. Let $1<a,b<5$. Again \eqref{eq:4.6} gives
\bEz
h_{15}\tht_{ab}(s)=0\,.
\eEz
It follows that non of the vectors $ \tht_{ab}(s)$ can be the hight weight
vector, they all are to be zero. Only $\tht_{1j}(s)$ for $j=2,3,4$ could be
non-zero.

For $b=3,4$ we have
\bEz
h_{25}\theta_{1b}(s)=
\theta_{12}(x_2\partial_b) (s)-(x_5\partial_b)\theta_{15} (s)-
(x_2\partial_1)\theta_{2b} (s)+(x_5\partial_1)\theta_{b5} (s)= 0
\,.
\eEz
So $\wt(\theta_{1b}(s))$ should be $(m,0,0,0)$ and this is not possible.
The only choice is $\theta_{1b}(s)=0$ for $b=3,4$.

We are left with the case when all $\tht_{ab}(s)=0$ except $\tht_{12}(s)$, which
is the highest weight vector. Then similarly we get $h_{35}\tht_{12}(s)=0$ and
therefore $\wt(\tht_{12}(s))=(m,n,0,0)$. Then
$\ph=\triangledown_A$. \epf


\section{Morphisms of degree 4.}
\label{sec:5}

Our goal is to construct a morphism of degree 4.

Let us start with the formula
\bEq\label{eq:5.1}
t=\sum u_{\bar{a}}\otimes \hat{z}^{\bar{a}}\,,
\eEq
where $\{\hat{z}^{\bar{a}}\}$ is the monomial basis of $\Sym^3({\CC^5}^*)$ and
$\{u_{\bar{a}}\}$ is the dual basis of the irreducible \\
$s\ell_5$-submodule
$\Sym^{3}({\CC}^{5})$ in $U_-$
with the highest vector $d_{12}d_{13}d_{14}d_{15}$.
This implies
\[
u_{(30000)}=d_{12}d_{13}d_{14}d_{15}\,.
\]
It is more convenient to write the multi-index in the "multiplicative" form,
$[1^3]$ instead of $(30000)$ and so on. Then
\bEq\label{eq:5.2}
u_{[1^2 2]}=u_{(21000)}=d_{12}d_{23}d_{14}d_{15}+
            d_{12}d_{13}d_{24}d_{15}+
               d_{12}d_{13}d_{14}d_{25}\,.
\eEq
The expressions of this type are not unique, for example
\bEz
\begin{aligned}
u_{[2^3]}&=d_{21}d_{23}d_{24}d_{25}=d_{12}d_{23}d_{24}d_{25}\,,\\
u_{[3^3]}&=d_{31}d_{32}d_{34}d_{35}=d_{13}d_{23}d_{34}d_{35}\,,\\
u_{[4^5]}&=d_{41}d_{42}d_{43}d_{45}=d_{14}d_{24}d_{34}d_{45}\,,
\end{aligned}
\eEz
but each $u_{\bar{a}}$ can be written as a sum of monomials in $d_{ij}$
of degree 4.

\bTh\label{th:deg4}
$(a)$ For each $n_1\geq 3$ there exist a morphism of degree $4$
\[
t_{AB}:M(n_1,0,0,0)\to M(n_1-3,0,0,0) \,.
\]
$(b)$ For any $n_4\geq 0$ there exist a morphism of degree $4$
\[
t_{BC}:M(0,0,0,n_4)\to M(0,0,0,n_4+3) \,.
\]
\eTh
\begin{remark}\label{rem:5.1}
We may also write $t_{AB}$ as a morphism of the combined Verma module
\[
M(\CC[z_i])=\oplus M(n,0,0,0)\, \text{ in itself }
,\quad t_{AB}:M(\CC[z_i])\to M(\CC[z_i])
\]
in itself given by \eqref{eq:5.1} with $\bar{z_i}=\partial_i$.
Similarly
\[
t_{BC}:M(\CC[z_i^*])\to M(\CC[z_i^*])
\]
is a morphism
given by \eqref{eq:5.1} with $\bar{z_i}$ equal to the multiplication by $z_i^*$.
\end{remark}
\pth.\/
We see immediately that Remark~\ref{rem:5.1} determines
a $(L_- \oplus \fg_0)$-morphism of Verma modules. One has to check the
condition \eqref{eq:1.4b} where we may suppose that $a$ is the highest weight
vector of $A$. For $t_{AB}$ this means
\bEz
\begin{aligned}
x_5d_{45}\cdot t_{AB}\otimes z_1^n
&=n(n-1)(n-2)\,x_5d_{45}\cdot d_{12}d_{13}d_{14}d_{15}\otimes z_1^{n-3}\\
   &=n(n-1)(n-2)\,((x_5\partial_3)d_{13}d_{14}d_{15}
   +d_{12}(x_5\partial_2)d_{14}d_{15})\otimes z_1^{n-3}=0\,.
\end{aligned}
\eEz

The case $(b)$ amounts to a much more complicated computation.
We shall write only few hints.
Evidently the symmetric group $\Sym_5$
acts  on the indices. The following simple observation helps in the computations.
\bLe For $\pi \in \Sym_5$ we have
$ \pi\cdot u_\al = \sgn(\pi)\,u_{\pi(\al)}.$
\eLe
Clearly we have the equalities
\[
u_{[1^2 2]}= (x_2\partial_1)\,u_{[1^3]}\quad\text{ and }\quad
u_{[1 2 3]}=(x_3\partial_1)\,u_{[1^2 2]}\,,
\]
that help us to write the elements $u_{\al}$ explicitly.
We leave it to the reader to check the rest of the computation.
\epf

\bPr Whenever defined, the following compositions are zero
\begin{align}\notag
&t_{AB}\cdot \triangledown_A =0\,,
\quad\triangledown_B \cdot t_{AB}=0\,,\\ \notag
&t_{BC}\cdot \triangledown_B =0\,,
\quad\triangledown_C \cdot  t_{BC}=0\,,\\
& t_{BC}\cdot t_{AB}=0\,.
\notag
\end{align}
\ePr
The proof is more or less immediate.\\

This means that we have got complexes that are shown in the following picture.


\begin{figure}[htbp]
  \begin{center}
    \leavevmode
  \setlength{\unitlength}{0.25in}
\begin{picture}(21,19)


\put(18,17.5){A:M(m,n,0,0)}
\put(18,-1){B:M(m,0,0,n)}

\put(1,17.5){}
\put(1,-1){C:M(0,0,m,n)}


\put(10,17.5){\line(0,-1){8.5} }
\put(9.9,18){n}

\put(10,-.5){\line(0,1){9.5} }
\put(9.9,-1){n}

\put(11.5,-.5){\line(0,1){8} }
\put(11.4,-1){n}

\put(11.5,7.5){\line(1,0){6.5} }
\put(18.5,7.4){m}

\put(1.8,9){\line(1,0){8.2} }
\put(1,8.9){m}

\put(10,9){\line(1,0){8} }
\put(18.5,8.9){m}


\thicklines
\multiput(2.5,0)(0,1.5){7}{\circle{.25} }
\multiput(4,0)(0,1.5){7}{\circle{.25} }  
\multiput(5.5,0)(0,1.5){7}{\circle{.25} }
\multiput(7,0)(0,1.5){7}{\circle{.25} }  
\multiput(8.5,0)(0,1.5){7}{\circle{.25} }

\multiput(10,0)(0,1.5){12}{\circle{.25} }  
\multiput(11.5,0)(0,1.5){6}{\circle{.25} }
\multiput(11.5,9)(0,1.5){6}{\circle{.25}} 
\multiput(13,0)(0,1.5){12}{\circle{.25} }  
\multiput(14.5,0)(0,1.5){12}{\circle{.25} }
\multiput(16,0)(0,1.5){12}{\circle{.25} }  
\multiput(17.5,0)(0,1.5){12}{\circle{.25} }

\thinlines

\drawline(1.95,0  )(2.35,0  )
\drawline(1.95,1.5)(2.35,1.5)
\drawline(1.95,3  )(2.35,3  )
\drawline(1.95,4.5)(2.35,4.5)
\drawline(1.95,6  )(2.35,6  )
\drawline(1.95,7.5)(2.35,7.5)
\drawline(1.95,9  )(2.35,9  )


\drawline(12.45,-.55)(12.9,-.1)
\drawline(13.95,-.55)(14.4,-.1)
\drawline(15.45,-.55)(15.9,-.1)
\drawline(16.95,-.55)(17.4,-.1)

\multiput(18.2,.7)(0,1.5){5}{\vector(-1,-1){.55}}
\multiput(10,17.25)(1.5,0){6}{\vector(0,-1){.55}}


\multiput(10,16.3)(1.5,0){6}{\vector(0,-1){1.05}}
\multiput(10,14.8)(1.5,0){6}{\vector(0,-1){1.05}}
\multiput(10,13.3)(1.5,0){6}{\vector(0,-1){1.05}}
\multiput(10,11.8)(1.5,0){6}{\vector(0,-1){1.05}}
\multiput(10,10.3)(1.5,0){6}{\vector(0,-1){1.05}}

\multiput(14.5,8.95)(1.5,0){3}{\vector(-2,-1){2.75}}
\multiput(18.15,7.8)(1.5,0){1}{\vector(-2,-1){0.5}}
\multiput(17.9,8.5)(1.5,0){1}{\vector(-2,-1){1.73}}

\multiput(3.7,9  )(1.5,0){5}{\vector(-1,0){0.9}}
\multiput(3.7,7.5)(1.5,0){5}{\vector(-1,0){0.9}}
\multiput(3.7,6  )(1.5,0){5}{\vector(-1,0){0.9}}
\multiput(3.7,4.5)(1.5,0){5}{\vector(-1,0){0.9}}
\multiput(3.7,3  )(1.5,0){5}{\vector(-1,0){0.9}}
\multiput(3.7,1.5)(1.5,0){5}{\vector(-1,0){0.9}}
\multiput(3.7,0  )(1.5,0){5}{\vector(-1,0){0.9}}

\multiput(12.85,7.35)(1.5,0){4}{\vector(-1,-1){1.2}}
\multiput(12.85,5.85)(1.5,0){4}{\vector(-1,-1){1.2}}
\multiput(12.85,4.35)(1.5,0){4}{\vector(-1,-1){1.2}}
\multiput(12.85,2.85)(1.5,0){4}{\vector(-1,-1){1.2}}
\multiput(12.85,1.35)(1.5,0){4}{\vector(-1,-1){1.2}}

\put(11.35,7.35){\vector(-1,-2){1.3}}
\put(11.35,5.85){\vector(-1,-2){1.3}}
\put(11.35,4.35){\vector(-1,-2){1.3}}
\put(11.35,2.85){\vector(-1,-2){1.3}}
\drawline(11.35,1.35)(10.45,-.45)
\drawline(11.35,-0.15)(11.15,-.55)
\end{picture}
\vspace{3ex}
  \caption{The long complexes with morphisms of degree 1 and 4.}
    \label{fig:2}
  \end{center}
\end{figure}
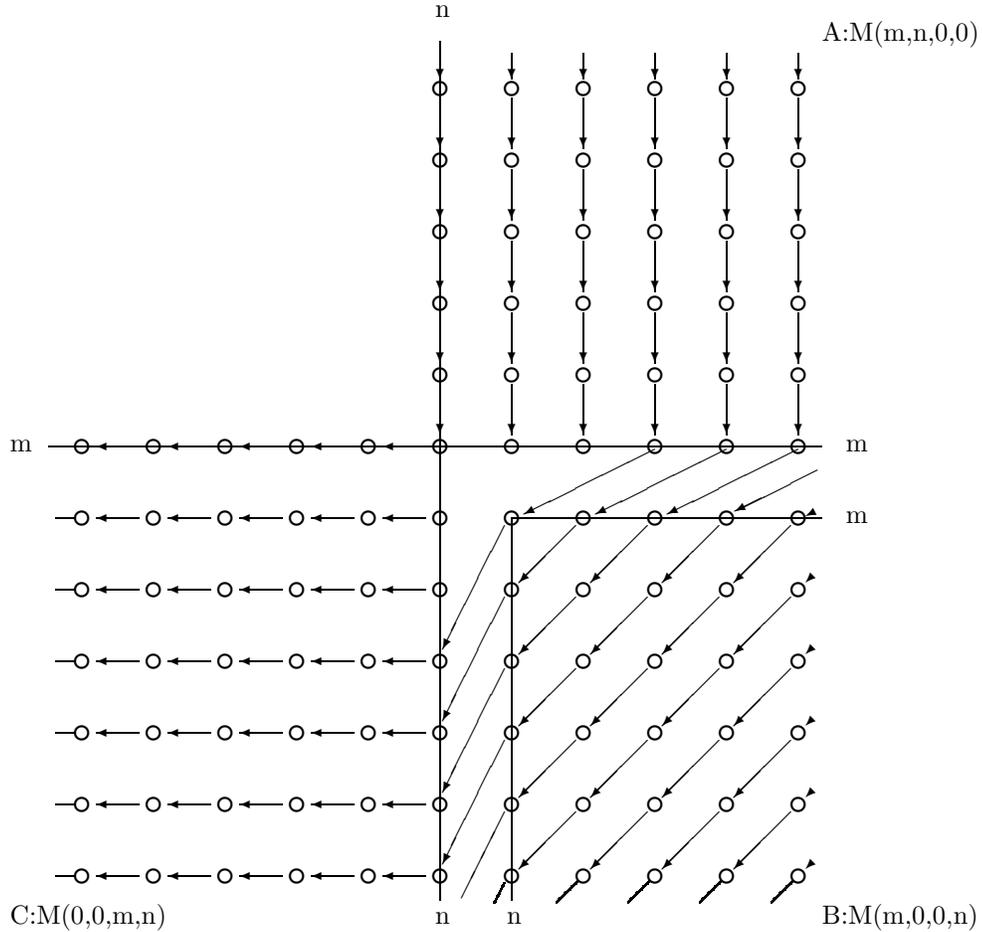


\section{Morphisms of other degrees.}
\label{sec:6}

There morphisms of degrees 2, 3 and 5 that can be constructed as compositions of
the morphisms of degree 1 and 4 described above.

\bPr\label{pr:deg2}
There are morphisms of degree $2$
\begin{align}\notag
\triangledown_{AB}&=\triangledown_B \!\cdot\! \triangledown_A :M(m,1,0,0) \to M(m-1,0,0,1)
\quad \text{ for $m>0$}\,,\\ \notag
\triangledown_{BC}&=\triangledown_C \!\cdot\! \triangledown_B
:M(1,0,0,n) \to M(0,0,1,n+1)
\quad \text{ for $n \geq 0$}\,,\\ \notag
\triangledown_{AC}&=
\triangledown_C \!\cdot\! \triangledown_A :M(1,0,0,0) \to M(0,0,0,1)\,.\notag
\end{align}
\ePr
It is not difficult to see that the morphisms are non-zero
and they are evidently morphisms of degree 2.
\bPr\label{pr:deg3}
There is a morphism of degree $3$
\[
\triangledown_{ABC}=
\triangledown_C \!\cdot\!\triangledown_B \!\cdot\! \triangledown_A
:M(1,1,0,0) \to M(0,0,1,1)\,.
\]
\ePr
It is a simple calculation to check that the morphism in question is non-zero.
\bPr\label{pr:deg5}
There are two morphisms of degree $5$
\begin{align}\notag
& t'  =
\triangledown_C \!\cdot\! t_{AB}
:M(3,0,0,0) \to M(0,0,1,0)\,\, \text{ and }&\\ \notag
&t'' =
 t_{BC} \!\cdot\!\triangledown_A
:M(0,1,0,0) \to M(0,0,0,3)\,.&
\end{align}
\ePr
Again we only need to notice that the morphisms are non-zero which amounts to
some calculation. It is tempting to believe that we have found already
all morphisms
between degenerate Verma modules.
\begin{Conjecture}
  \label{conj:6.1}
The morphisms listed in Propositions~\ref{pr:deg2},~\ref{pr:deg3},~\ref{pr:deg5}
and Theorems~\ref{th:4.2},~\ref{th:deg4} are all morphisms between degenerate (minimal) Verma
modules for $E(5,10)$, in particular there are no morphisms of degrees larger
than $5$.
\end{Conjecture}

\vspace{6ex}

\textbf{Author addresses:}

\begin{list}{}{}

\item  NIISI RAN, section Mathematics,\\
Nahimovskij prosp.36, kor.1,\\
 Moscow 117218 RUSSIA\\
\item Dept of Mathematics, GU-VSE, \\
Moscow, RUSSIA \\
\item email:~~rrudakov@gmail.com\,

\end{list}

\end{document}